\documentclass[aps,prl,twocolumn,showpacs,superscriptaddress,groupedaddress]{revtex4}
\usepackage{amsmath}
\usepackage{amsfonts}
\usepackage{amssymb}
\usepackage{graphicx}
\usepackage{morefloats}
 \usepackage{dcolumn}   
\usepackage{bm}       
\usepackage{multirow}
\usepackage{relsize}   
\usepackage{amsthm}

\begin{document}

\newcommand{\ket}[1]{\left |#1 \right \rangle}
\newcommand{\bra}[1]{\left  \langle #1 \right |}
\newcommand{\braket}[2]{\left \langle #1 \middle | #2 \right \rangle}
\newcommand{\ketbra}[2]{\left | #1 \middle \rangle \middle \langle #2 \right |}

\newtheorem{conj}{Conjecture}
\newtheorem{thm}{Theorem}

\widetext

\title{Pairwise Concurrence in Cyclically Symmetric Quantum States}
\author{Alexander Meill and David A. Meyer}
\date{\today}

\begin{abstract}
We provide an initial characterization of pairwise concurrence in quantum states which are invariant under cyclic permutations of party labeling.  We prove that maximal entanglement can be entirely described by adjacent pairs, then give explicit descriptions of those states in specific subsets of 4 and 5 qubit states - X states.  We also construct a monogamy bound on shared concurrences in the same subsets in 4 and 5 qubits, finding that above non-maximal entanglement thresholds, no other entanglements are possible.
\end{abstract}

\maketitle

\section{Introduction}

Entanglement in quantum mechanics has been an exciting avenue of research in physics since its discovery.  It plays a central role in quantum computing \cite{1}\cite{2} and offers meaningful contributions to high energy theory \cite{3} and condensed matter physics \cite{4}\cite{5}.  Despite the attention that entanglement has received, its fundamental properties are still not fully understood.  Constraints on entanglement are generally challenging to compute due to the fact that many entanglement measures involve extremizations which are difficult to handle analytically.  Those measures which do have a closed function on state parameters are difficult to calculate for high dimensional systems and many particles.

A common approach to studying these large Hilbert spaces is to consider entanglement in some smaller subspace which reduces the number of state parameters.  Entanglement has been studied in states which are invariant under permutation of party labeling \cite{6}, X-states \cite{7}, and matrix product states \cite{8} among other subsets.  This paper considers the pairwise concurrence entanglement measure, defined in \cite{9}, of $n$ qubit states which are invariant under cyclic permutation of party labeling.  These cyclically symmetric (CS) states are of significant interest to translation-invariant condensed matter systems \cite{20}\cite{19}\cite{10} and 1-D spin chains with periodic boundary conditions \cite{11}.  Their SLOCC properties were also examined in \cite{21}.

The CS subspace of an $n$ qubit system offers a significant simplification to the entanglement picture by constraining the number of allowable distinct types of entanglement.  Any subset or partitioning of parties to calculate entanglement among, no matter the measure of entanglement, would be equated to that of other sets of parties by the cyclic permutation invariance of the state.  We narrow this picture by only examining pairwise entanglements as measured by the concurrence, $\mathcal C$, which is chosen for its relative analytic symplicity and for its relationship to the entanglment of formation \cite{19}.  The cyclic symmetry implies that for any pairwise concurrence $\mathcal{C}_{i,j}$ between parties $i$ and $j$, $\mathcal{C}_{i,j}=\mathcal{C}_{i+k,j+k}$, where the party label subscripts are to be evaluated mod $n$.  So each allowable pairwise concurrence in CS-states corresponds to the spacing between party labelings.  As a point of notation, define $\mathcal{C}_k^{(n)}$ to be the pairwise concurrence between parties $k$-away in an $n$ qubit CS-state.  Note that $k$ runs from 1 to $\lfloor \frac n2 \rfloor$ as any $k > \frac n2$ is equivalent to the $n-k$ spacing.  The $\lfloor \frac n2 \rfloor$ distinct $\mathcal{C}_k^{(n)}$ is reduced from the ${\binom n2}$ distinct pairs in a general $n$ qubit state.

The entanglement picture in CS-states is further simplified by the fact that many $\mathcal{C}_k^{(n)}$ share the same properties.  To see this, consider some $m$ which is not a factor of $n$, and the associated permutation, $\pi \in S_n$,
\begin{equation}
\pi: i \mapsto  m i \mod n.
\end{equation}
Note that $\pi$ is invertible only when $m=1$ or $m \nmid n$.  Where obvious, we will interchangably use $\pi$ to denote the permutation on the tensor factors, as well as the associated unitary operator acting on the state.  Permuting the party labels of some CS-state, $\ket{\psi}$, according to $\pi^{-1}$ will leave the state in some new CS-state, $\ket{\chi}=\pi^{-1} \ket{\psi}$, which obeys $\mathcal{C}_{i,j} \left( \ket{\psi} \right) = \mathcal{C}_{\pi(i),\pi(j)} \left( \ket{\chi} \right)$.  This means that any properties of $\mathcal{C}_k^{(n)}$ will be shared by $\mathcal{C}_{mk}^{(n)}$ for each $m$ which is not a factor of $n$.  It then suffices to only examine the constraints on $\mathcal{C}_k^{(n)}$ for $k | n$.

These simplifications, along with the natural reduction in state parameters, makes an analytic description of the CS entanglement more approachable.  This paper makes a preliminary attempt at analyzing the allowed pairwise concurrences in CS-states.  First, we prove that maximal entanglement in CS-states can be entirely understood in terms of the maxima of $\mathcal C_1^{(n)}$ and explicitly determine the maxima on the X-state subspace for 4 and 5 qubits.  We then discuss the bounds on multiple concurrences, again with an analytic description for X-states in 4 and 5 qubits.  Due to the extensive nature of the calculations, significant portions of analysis are relegated to the appendices.

\section{Maximally Entangled States}

A natural question when examining a subset of quantum states is which states maximize entanglement within that subset, and what is that maximal entanglement?  As a result of the discussion in the previous section, we need only examine the maxima of $\mathcal{C}_k^{(n)}$ for $1\leq k \leq \lfloor \frac n2 \rfloor$ which are not factors of $n$.  Denote a state which maximizes $\mathcal{C}_k^{(n)}$ as $\ket{\psi_k^{(n)}}$.  Finding the $\ket{\psi_k^{(n)}}$ and the associated maximal $\mathcal{C}_k^{(n)}$ is greatly simplified by the following theorem,

\begin{thm}
For $k|n$, $\max \mathcal{C}_k^{(n)}= \max \mathcal{C}_1^{(n/k)}$ and a corresponding state which maximizes $ \mathcal{C}_k^{(n)}$ can be constructed as 
\begin{equation}
\ket{\psi_k^{(n)}}= \bigotimes_{i=0}^{k-1} \ket{\psi_1^{(n/k)}}_{ k \{n/k\} + i },
\end{equation}
where $\{n/k\}$ represents the set of integers from $0$ to $n/k-1$.  These integers, multiplied by $k$ then incremented by $i$, indicate the party labelings in the overall state.
\end{thm}
\begin{proof}
Consider some $n$ qubit CS-state, $\ket{\psi^{(n)}}=\sum_{\textbf{i}\in \mathbb Z_2^n} \psi_{\textbf i} \ket{\textbf i}$ and some $k|n$.  Examine the reduced density matrix,
\begin{eqnarray}
\rho_{k \{n/k\}} &=& \text{Tr}_{\overline{k \{n/k\}}} \left( \ketbra{\psi^{(n)}}{\psi^{(n)}} \right) \\
&=& \sum_{\textbf j \in \mathbb Z_2^{n-n/k}} \sum_{\textbf a, \textbf b \in \mathbb Z_2^{n/k}} \psi_{\textbf j, \textbf a}^{\vphantom{*}} \ketbra{\textbf a}{\textbf b} \psi_{\textbf j, \textbf b}^*,
\end{eqnarray}
where $\textbf a$ and $\textbf b$ indicate basis elements in the parties in $k \{ n/k \}$, while $\textbf j$ indicate basis elements in the remaining $n-n/k$ parties.  Notably, this reduced state obeys, by definition, $\mathcal{C}_1^{(n/k)} \left( \rho_{k \{n/k\}} \right) = \mathcal{C}_k^{(n)} \left( \ket{\psi^{(n)}} \right )$.

Now label any $\pi \in \mathbb Z_{n} \subset S_n$ as
\begin{equation}
\pi_m^{(n)}:i \mapsto i+m \mod n.
\end{equation}
We can then examine that, for any $m$,
\begin{eqnarray}
\pi_{n/k-m}^{(n/k)} \rho_{k \{ n/k\}}^{\vphantom{(}} &=&  \sum_{\textbf j } \sum_{\textbf a, \textbf b} \psi_{\textbf j, \pi_m^{(n/k)} (\textbf a)}^{\vphantom{*}} \ketbra{\textbf a}{\textbf b} \psi_{\textbf j, \textbf b \vphantom{\pi_m^{(n/k)}}}^* \\
 &=& \sum_{\textbf j } \sum_{\textbf a, \textbf b} \psi_{\pi_{km}^{(n)} (\textbf j, \textbf a)}^{\vphantom{*}} \ketbra{\textbf a}{\textbf b} \psi_{\textbf j, \textbf b \vphantom{\pi_m^{(n/k)}}}^* \\
&=& \rho_{k \{n/k \}},
\end{eqnarray}
where the first equality describes the action of a permutation on the parties in $k\{n/k\}$, the second extends that permutation to the $n$ parties and rearranges using the sum over \textbf{j}, and the third uses the cyclic symmetry of $\ket{\psi^{(n)}}$.
And so, for any $\pi \in \mathbb Z_{n/k}$,
\begin{equation}
\pi \rho_{k \{ n/k \}} = \rho_{k \{n/k\}} \pi = \rho_{k \{n/k\}}.
\end{equation}
Since $\rho_{k\{n/k\}}$ commutes with $\pi_1^{(n/k)}$, they can be simultaneously diagonalized into a basis $\left \{ \ket{\phi_j} \right \}$.  Since $\pi_1^{(n/k)}$ is unitary, its eigenvalues associated to each $\ket{\phi_j}$, can be labeled as $\lambda_j = e^{i \phi_j}$.  We can then examine
\begin{eqnarray}
\pi_{1\vphantom{\{}}^{(n/k)} \rho_{k \{n/k\}}^{\vphantom (} &=& \sum_j p_j^{\vphantom (} \pi_1^{(n/k)} \ketbra{\phi_j}{\phi_j} \\
&=& \sum_j p_j e^{i \phi_j} \ketbra{\phi_j}{\phi_j},
\end{eqnarray}
which, according to equation (9), must be equal to the original $\rho_{k\{n/k\}}$.  This is only possible if $e^{i \phi_j}=1$ for each $j$, implying that $\ket{\phi_j}$ are each CS-states.

Lastly, order the eigenstates to be decreasing in $\mathcal C_1^{(n/k)} \left( \ket{\phi_j} \right)$.  By the convexity of the pairwise concurrence, it then follows that
\begin{eqnarray}
\mathcal C_{1\vphantom{\{}}^{(n/k)} \left( \rho_{k \{n/k\}}^{\vphantom (} \right) &=& \mathcal C_1^{(n/k)} \left( \sum_j p_j \ketbra{\phi_j}{\phi_j} \right) \\
&\leq& \sum_j p_j^{\vphantom (} \mathcal C_1^{(n/k)} \left( \ket{\phi_j} \right) \\
& \leq & \mathcal C_1^{(n/k)} \left( \ket{\phi_1}\right) \\
& \leq & \mathcal C_1^{(n/k)} \left( \ket{\psi_1^{(n/k)}} \right),
\end{eqnarray}
with the inequality being saturated by the state, (2).
\end{proof}
Interestingly, convexity was the only property of the concurrence used in the proof of Theorem 1, meaning that any entanglement measure would obey an analagous statement in CS-states.

Notably, (2) also agrees with the monogamy behavior examined in the next section, as each of $\mathcal{C}_{j \neq k}^{(n)} \left( \ket{\psi_k^{(n)}} \right)=0$.  As a result of Theorem 1, all that remains is to find $\mathcal{C}_1^{(n)}$ for each $n$.  For $n \leq 3$, the CS subspace is equivalent to the totally symmetric one, where the maxima have previously been determined.  This leads to $\max \mathcal{C}_1^{(2)}=1$ with $\ket{\psi_1^{(2)}}= \frac 1 {\sqrt2}( \ket{00} + \ket{11} )$ and $\max \mathcal{C}_1^{(3)}=\frac 23$ with $\ket{\psi_1^{(3)}}= \frac 1 {\sqrt3}( \ket{001} + \ket{010} + \ket{100} )$ \cite{18}.  Turning to the $n\geq4$ case, some notation needs to be established.  Recall the Dicke basis \cite{17} element for totally symmetric states,
\begin{equation}
\ket{S_j^{(N)}} = {\binom nj}^{-1/2} \sum_{\pi \in S_n} \pi | \underbrace{00...0}_{n-j}\underbrace{11...1}_{j} \rangle,
\end{equation}
where the sum runs over all party label permutations.  This naturally extends to a CS basis element in the following manner.  For any particular computational basis element, a CS-state must have the same coefficient for each cyclic permutation of that basis element.  Let a normalized $n$ qubit CS basis element be denoted with an overbrace,
\begin{equation}
\overbrace{\ket{i_1 i_2 ... i_n}} = |\mathbb{Z}_n \ket{i_1 i_2 ... i_n} |^{-\frac 12} \sum_{\pi \in \mathbb{Z}_n} \pi \ket{i_1 i_2 ... i_n},
\end{equation}
where $|\mathbb{Z}_n \ket{i_1 i_2 ... i_n} |$ denotes the cardinality of the orbit of $\ket{i_1 i_2 ... i_n}$ under the action of the $\mathbb Z_n$ cyclic permutation group.  For example, consider the 4 qubit basis element,
\begin{equation}
\overbrace{\ket{0011}} = \frac 1 2 \biggr{[} \ket{0011} + \ket{1001} + \ket{1100}+\ket{0110} \biggr{]}.
\end{equation}

Using this basis notation, an arbitrary 4 qubit CS-state takes the form, 
\begin{equation}
\begin{split}
\ket{\psi^{(4)}} = &a \ket{0000} + b \overbrace{\ket{0001}} + c \overbrace{\ket{0011}}\\+  &d \overbrace{\ket{0101}} + e \overbrace{\ket{0111}} + f \ket{1111},
\end{split}
\end{equation}
where $|a|^2 + |b|^2 +|c|^2+|d|^2+|e|^2+|f|^2 = 1$.  Likewise, an arbitrary 5 qubit CS-state would be
\begin{equation}
\begin{split}
\ket{\psi^{(5)}} = &a \ket{00000} + b \overbrace{\ket{00001}} + c \overbrace{\ket{00011}} +  d \overbrace{\ket{00101}} \\ + &e \overbrace{\ket{00111}} + f \overbrace{\ket{01011}} + g \overbrace{\ket{01111}} + h \ket{11111},
\end{split}
\end{equation}
with the corresponding normalization.  Unfortunately, even calculating $\mathcal{C}_1^{(4)}$ and $\mathcal{C}_1^{(5)}$ for arbitrary states is analytically challenging, let alone maximizing over that space.  Instead, the calculation will be performed on the even-X-state subspaces for $n=4$ and $n=5$.  Even-X-states (abbreviated X-states), introduced in \cite{13}, are superpositions of only computational basis elements containing an even number of `1' entries.  Notably, the set of CS-states examined in \cite{20} are a subset of the CSX-states.  Arbitrary 4 and 5 qubit CSX-states then take the form,
\begin{eqnarray}
\ket{\psi_X^{(4)}} = a \ket{0000} + c \overbrace{\ket{0011}}+  d \overbrace{\ket{0101}} + f \ket{1111}, \quad \, \, \, \, \\
\ket{\psi_X^{(5)}} = a \ket{00000}  + c \overbrace{\ket{00011}} +  d \overbrace{\ket{00101}} + g \overbrace{\ket{01111}}. 
\end{eqnarray}
The X-state subspace is a useful one as concurrence calculations on the space are rather simple.  Two qubit reduced density matrices of X-states were shown in \cite{13} to be of the form
\begin{equation}
\rho = \left(\begin{array}{cccc}\alpha & 0 & 0 & \nu \\0 & \beta & \mu & 0 \\0 & \mu^* & \gamma & 0 \\\nu^* & 0 & 0 & \delta\end{array}\right).
\end{equation}
The square roots of the eigenvalues of $\rho \tilde{\rho}$ (as in the concurrence definition) \cite{9} are the following,
\begin{equation}
\lambda_i = \left \{ \sqrt{\beta \gamma} + |\mu |, \, \sqrt{\beta \gamma} - |\mu | , \, \sqrt{\alpha \delta} + |\nu | , \, \sqrt{\alpha \delta} - |\nu | \right \}.
\end{equation}
Either the first or third term is the largest eigenvalue so the X-state concurrence is then
\begin{equation}
\mathcal{C}(\ket{\psi_X}) = 2\max \left \{0,|\nu|-\sqrt{\beta \gamma},|\mu|-\sqrt{\alpha \delta} \right \}.
\end{equation}
Let $\mathcal{C}_{k,\mu}^{(n)}$ and $\mathcal{C}_{k,\nu}^{(n)}$ indicate the possible non-zero expressions for CSX concurrence involving $\mu$ and $\nu$ respectively.  Following this notation, the concurrences of arbitrary 4 and 5-qubit CSX-states can be calculated to be,
\begin{widetext}
\begin{eqnarray}
\mathcal{C}_{1,\mu}^{(4)} &=&\frac{|cd^*+dc^*|}{ \sqrt{2}} - 2\sqrt{ \left( |a|^2 + \frac{|c|^2}{4} \right) \left( \frac{ |c|^2}{4}+|f|^2 \right)} \\
\mathcal{C}_{1,\nu}^{(4)} &=& |ac^*+cf^*| - \frac{1}{2}|c|^2-|d|^2 \\ 
\mathcal{C}_{2,\mu}^{(4)} &=& |c|^2 - 2\sqrt{ \left( |a|^2+\frac{|d|^2}{2} \right) \left (\frac{|d|^2}{2}+|f|^2 \right)} \\
\mathcal{C}_{2,\nu}^{(4)} &=&\sqrt{2} |ad^*+df^*| - |c|^2 \\
\mathcal{C}_{1,\mu}^{(5)} &=& \frac 25\left(\left|dc^*+cd^*\right|+|d|^2+|g|^2 - \sqrt{\left(5|a|^2+2|c|^2+|d|^2 \right) \left( |c|^2+3|g|^2\right)}\right) \\
\mathcal{C}_{1,\nu}^{(5)} &=&\frac 25\left(\left| \sqrt{5}ac^*+2cg^* + dg^* \right| - |c|^2-2|d|^2-|g|^2 \right) \\
\mathcal{C}_{2,\mu}^{(5)} &=&\frac 25 \left(\left|dc^*+cd^*\right|+|c|^2+|g|^2 - \sqrt{\left(5|a|^2+|c|^2+2|d|^2 \right) \left( |d|^2+3|g|^2\right)}\right) \\
\mathcal{C}_{2,\nu}^{(5)} &=&\frac 25 \left(\left| \sqrt{5}ad^*+2dg^* + cg^* \right| - 2|c|^2-|d|^2-|g|^2\right).
\end{eqnarray}
\end{widetext}
In determining the maximum of $\mathcal{C}_1^{(4)}$ and $\mathcal{C}_1^{(5)}$ over the X-state subspace, the maximization will need to be performed over both the $\mu$ and $\nu$ terms, with the overall maximum being the larger of the two resulting maxima.  These maximizations are easily performed after setting all the coefficient phases equal to 0.  This phase treatment maximizes each absolute value in equations (13)-(20) and simplifies the maximizations enough to readily calculate.  The results are compiled in the table below.
\begin{table}[h]
\begin{center}
\renewcommand{\arraystretch}{2}
\begin{tabular}{| c | c |}
\hline
Concurrence & Maximum \\ \hline \hline
$\mathcal{C}_{1,\mu}^{(4)}$ & $\frac 14$  \\  \hline
$\mathcal{C}_{1,\nu}^{(4)}$ & $\frac 12$  \\  \hline
$\mathcal{C}_{1,\mu}^{(5)}$ & $\approx 0.468$  \\  \hline
$\mathcal{C}_{1,\nu}^{(5)}$ & $\approx 0.366$  \\  \hline
\end{tabular}
\caption{Maximum concurrences of 4 and 5 qubit CSX-states.  The analytic results for $n=5$ are the roots of complicated polynomials, so their rounded numerical values are reported instead.}
\end{center}
\label{default}
\end{table}
The overall maximum of $\mathcal{C}_{1}^{(4)}=\frac 12$ occurs when $d=0$ and $a=c=f=\frac 1 {\sqrt3}$, while the $\mathcal{C}_{1}^{(5)}\approx 0.468$ maximum occurs at  $a=g=0$ and $c \approx 0.298$ $d \approx 0.955$.  These maxima, while calculated only over the CSX subspace, agree with the apparent maxima in numerical results for general CS-states as shown in Figure 3 in the next section.  This $\mathcal C_1^{(5)}$ maximum is also a notable improvement over the lower bound established in \cite{20}.

For $n>5$, the CSX-state concurrences can be calculated, but the spaces prove too large and complicated to maximize over analytically.

\section{Constraints on Shared Entanglement}

The space of allowable pairwise concurrences, $\{ \mathcal{C}_{i,j} \}$ with $i$ from 1 to $n-1$ and $j>i$, for a general $n$ qubit state is known to be constrained by monogamy relations \cite{12}.  The pairs of $\{ \mathcal{C}_k^{(n)} \}$ for CS-states obey constraints of a similar nature.  Shown in Figure 1 are the $k=1$ and $k=2$ concurrences for $10^5$ randomly generated 4 and 5 qubit CS-states.  Note that the 5 qubit concurrence space is symmetric due to the permutation properties discussed in the introduction.
\begin{widetext}
\begin{center}
\begin{figure}[h]
\centering
\includegraphics[width = 150mm]{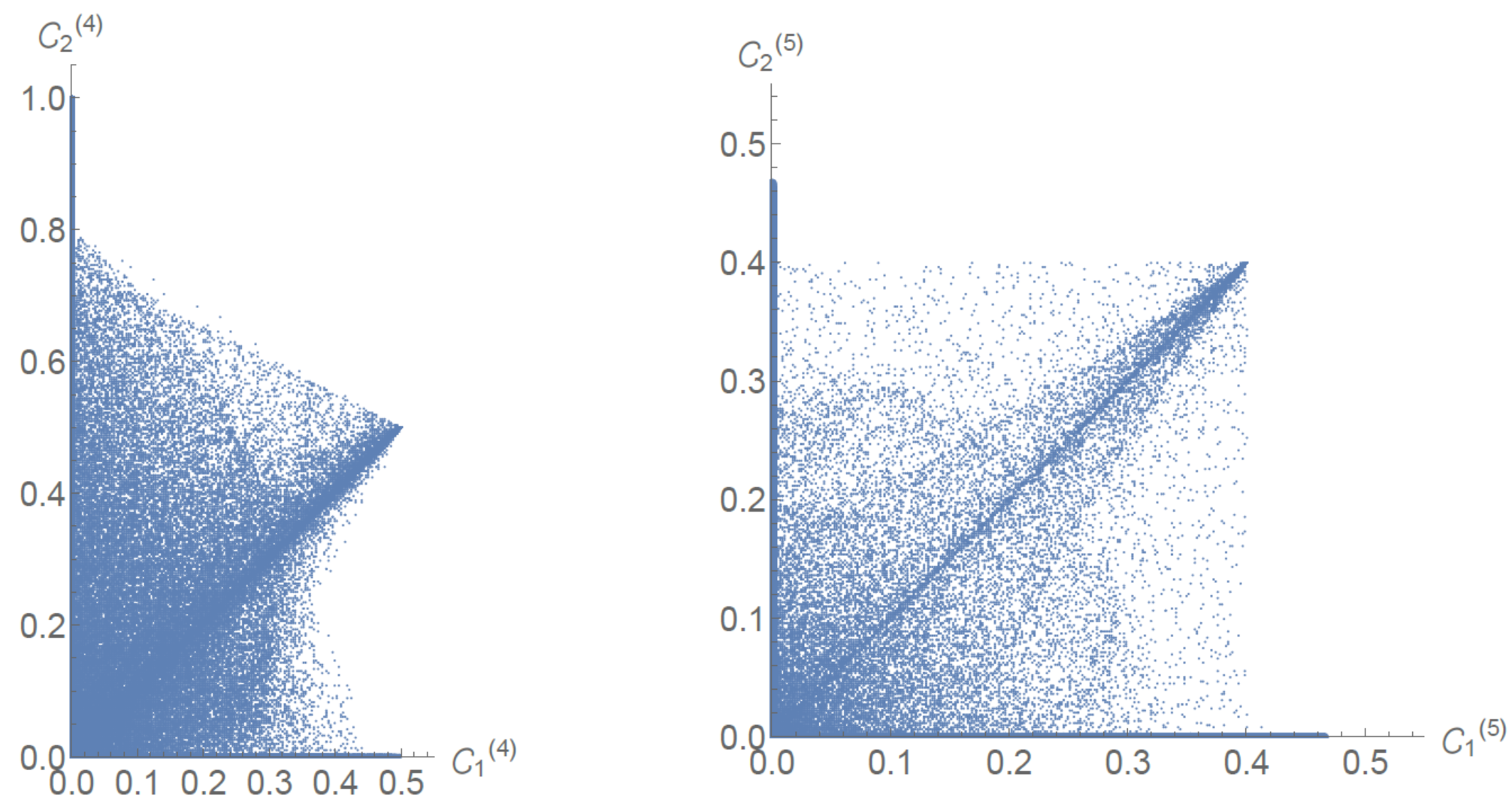}
\caption{ Pairwise concurrences of $10^5$ randomly generated 4 and 5 qubit CS states.}
\end{figure}
\end{center}
\end{widetext}
\noindent This first examination demonstrates the peculiar monogamous relationship between pairwise concurrences in CS-states.  It appears that for both $n=4$ and $n=5$, above some threshold concurrence, the other concurrence must be equal to 0.  This is differs from typical monogamy relations \cite{12}\cite{14}, which also suggest that the maximally entangled states minimize entanglement with other parties, but that states with slightly less entanglement than the maximum may share other entanglements.

The following theorem provides some analytical context to the CS-state monogamy.
\begin{thm}
The neighborhood of states around any $\ket{\psi_2^{(4)}}$ have $\mathcal{C}_1^{(4)}=0$.
\end{thm}
\begin{proof}
Consider the state,
\begin{equation}
\ket{\psi_2^{(4)}} = \ket{\psi_1^{(2)}}_{1,3} \otimes \ket{\psi_1^{(2)}}_{2,4},
\end{equation}
The pure 2 qubit states with concurrence equal to 1 are equivalent to each other under local unitaries, so the set of ${\ket{\psi_2^{(4)}}}$ are likewise equivalent.  This implies that the entanglement properties of any $\ket{\psi_2^{(4)}}$ can be determined by exmining those of (34).  Now consider altering (34) by some infinitesimal perturbation of the form of (19),
\begin{equation}
\ket{\psi '} = \ket{\psi_2^{(4)}} + \epsilon \ket{\psi^{(4)}},
\end{equation}
where $\epsilon \ll1$.  To show $\mathcal{C}_{1}^{(4)}=0$ for the above state regardless of the perturbation, we first calculate the reduced density matrix between adjacent parties,
\begin{equation}
\rho_r = \frac {\mathbb{1}}4 + \frac {\epsilon}2  \Re \left [ \left(\begin{array}{cccc}2a & b & b & c \\ b & \sqrt{2}d & c & e \\b & c & \sqrt{2}d & e \\ c & e & e & 2 f\end{array}\right) \right] + \mathcal O (\epsilon^2).
\end{equation}
It's clear that only the real part of the perturbation will affect the concurrence, so continue assuming the coefficients of the perturbation are real.  For simplicity, absorb $\epsilon$ into the perturbation coefficients.  Continuing in the concurrence calculation,
\begin{equation}
\rho_r \tilde{\rho_r} = \frac {\mathbb1} {16} + \frac 18 \left(\begin{array}{cccc}2a+2f & b-e & b-e & 2c \\ b-e & 2\sqrt{2}d & 2c & e-b \\b-e & 2c & 2\sqrt{2}d & e-b \\ 2c & e-b & e-b & 2a+2 f\end{array}\right) + \mathcal O (\epsilon^2).
\end{equation}
The square roots of the eigenvalues of this matrix are all $\lambda_i = \frac 14 \sqrt{1 + \mathcal O (\epsilon) + \mathcal O (\epsilon^2)}$.  Therefore, the sum $\lambda_1-\lambda_2-\lambda_3-\lambda_4$ will certainly be negative, so the concurrence is 0.
\end{proof}

The monogamy of CS-states is more clearly observed by examining the subconcurrence, defined as
\begin{equation}
s\mathcal{C} = \lambda_1-\lambda_2-\lambda_3-\lambda_4,
\end{equation}
where $\lambda_i$ are the square roots of the eigenvalues of $\rho \tilde{\rho}$ in descending magnitude, as in the concurrence definition.  More simply, the subconcurrence has the same definition as the concurrence, except it doesn't map negative sums of $\lambda_i$ to 0.  The subconcurrences of randomly generated 4 and 5 qubit CS-states are displayed in Figure 2.
\begin{widetext}
\begin{center}
\begin{figure}[h]
\centering
\includegraphics[width = 150mm]{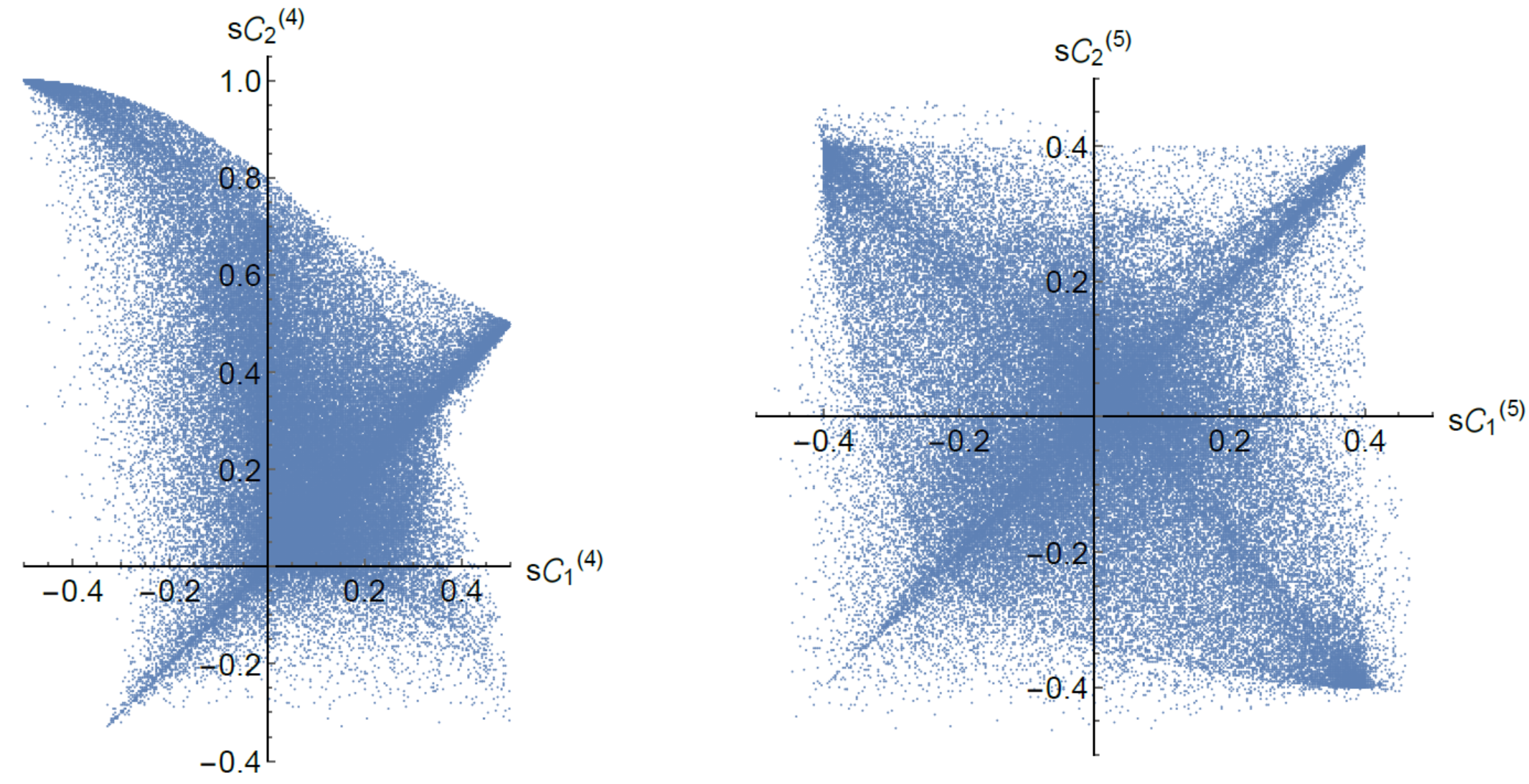}
\caption{Pairwise subconcurrences of $10^5$ randomly generated 4 and 5 qubit CS states.}
\end{figure}
\end{center}
\end{widetext}

Figure 2 clearly demonstrates the apparent thresholds in 4 and 5 qubits.  For both $n=4$ and $n=5$, it appears that above some $k=2$ subconcurrence, the $k=1$ subconcurrence must be negative.  Due to the symmetry discussed in the introduction, in 5-qubits, states with $k=1$ subconcurrences above the same threshold will have negative $k=2$ subconcurrence.  For $n=4$, however, the totally symmetric state, $\ket{W} = \overbrace{\ket{0001}}$ has the same $s\mathcal{C}_1^{(4)}$ as (34) while also having $s\mathcal{C}_2^{(4)}=\frac 12$.

The analytic description of these monogamy thresholds will again be performed on the X-state subspace, where the calculations are much simpler.  Shown in Figure 3 are the subconcurrences of randomly generated CSX-states overlaid on general CS-state subconcurrences.  Based on these numeric results, it is apparent that CSX-states share the same monogamy thresholds and maximum concurrences as CS-states, making them a relevant subset for analysis.
\begin{widetext}
\begin{center}
\begin{figure}[h]
\centering
\includegraphics[width = 150mm]{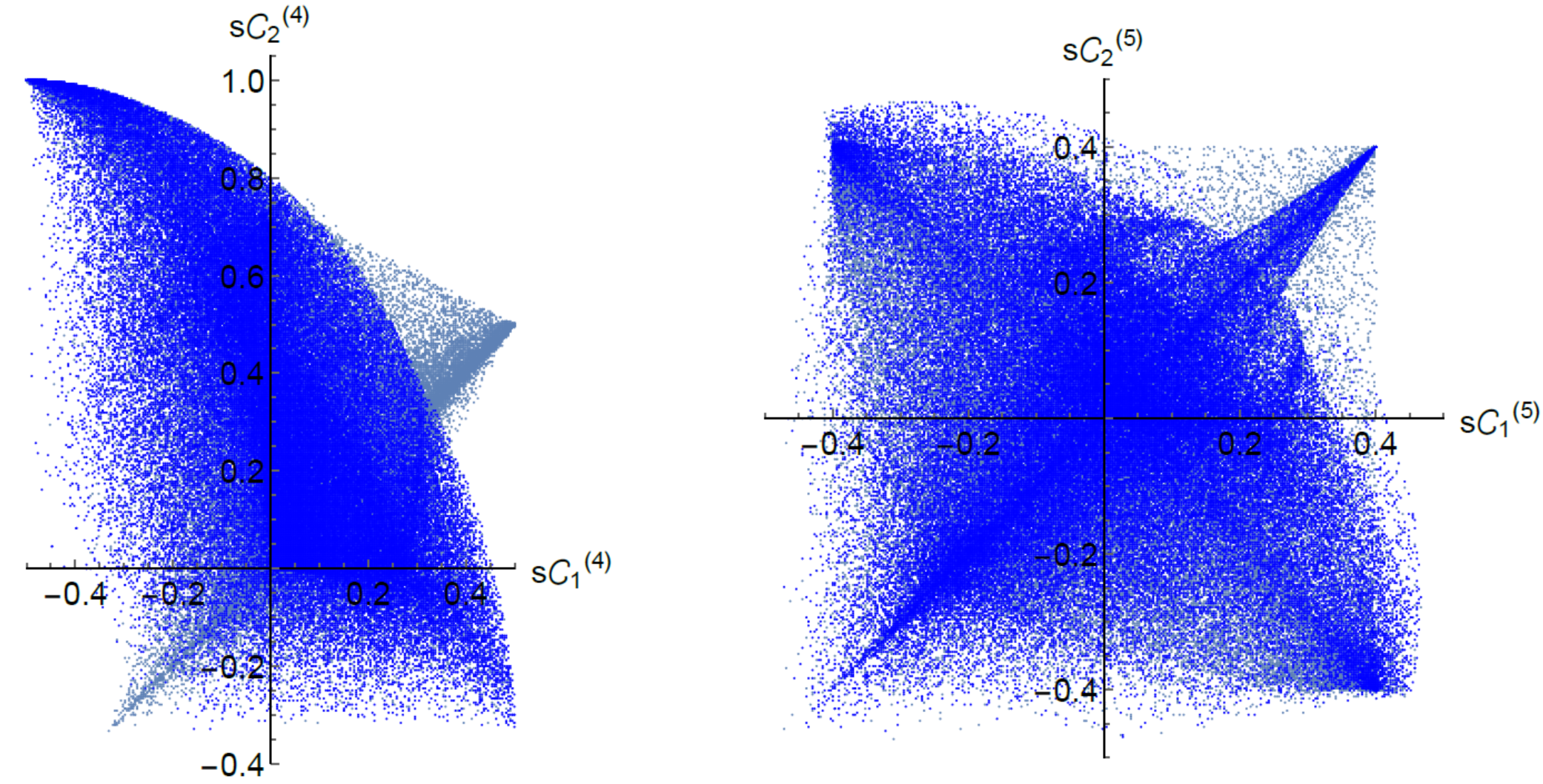}
\caption{Pairwise subconcurrences of $10^5$ randomly generated 4 and 5 qubit CS (blue) and CSX (darker blue) states.}
\end{figure}
\end{center}
\end{widetext}

Looking only at CSX-states, we found the acheivable concurrence boundaries in both 4 and 5 qubits.  The full analysis is presented in the appendix, but the boundaries allow for a quick determination of the concurrence thresholds in the X state subspace.  The thresholds are compiled in Table II on the next page.
\begin{table}[h]
\begin{center}
\renewcommand{\arraystretch}{2}
\begin{tabular}{| c | c |}
\hline
Concurrence & Threshold \\ \hline \hline
$s\mathcal{C}_{1}^{(4)}$ & $\frac {2\sqrt 2 -1}4$  \\  \hline
$s\mathcal{C}_{2}^{(4)}$ & $\frac 45$  \\  \hline
$s\mathcal{C}_{k}^{(5)}$ & $ \approx 0.418$  \\  \hline
\end{tabular}
\caption{Threshold concurrences of 4 and 5 qubit CSX-states.  The analytic result for $n=5$ is the root of a complicated polynomial, so the rounded numerical value is reported instead.}
\end{center}
\label{default}
\end{table}
Note that the $s\mathcal{C}_{1}^{(4)}$ threshold only fully holds for CSX-states.  Also recall that the concurrence symmetry in 5 qubits implies that $s\mathcal{C}_{1}^{(5)}$ and $s\mathcal{C}_{2}^{(5)}$ have the same threshold.

\section{Discussion}

In the search for maximally entangled state in $n$ qubit CS-states we have provided a state construction which reduces the problem to finding the states which maximize the concurrence between adjacent parties.  Adjacent maxima are well understood in 2 and 3 qubits, and we have calculated the maximum for 4 and 5 qubits in the X-state subspace.  Brute force calculations are obviously difficult in larger $n$, even in the X-state subspace.  The development of a generalized basis for large $n$ qubit CS-states, similar to the Dicke basis for totally symmetric states, would possibly enable more general statements without quite as much raw calculation.  In addition, a canonical form resulting from local unitaries which leaves the state in some simpler, yet still cyclically symmetric, state would aid in calculation.  Presently, no such canonical form is known for CS-states.

The work of this paper would be interesting and simple enough to repeat with alternate pairwise entanglement measures, such as the Negativity.  In particular, Theorem 1 would still hold for the Negativity and would make conclusions about adjacent entanglement equally generalizable.  It would also be simple enough to extend Theorem 1 and the other entanglement permutation relations to non-pairwise measures such as the 3-tangle or bipartite entanglement between bipartitions of the parties in the overall state.

This work was supported, in part, by NSF grant PHY-1620846.

\appendix

\section{Appendix:  X State Achievable Subconcurrence Boundaries}

To find the boundary of CSX-state subconcurrences, the boundaries of each of the pairs $\left( s\mathcal C^{(n)}_{1,\mu(\nu)}, s\mathcal C^{(n)}_{2,\mu(\nu)} \right)$ need be found, with the overall boundary being a combination of the outermost boundaries from each pairing due to $s \mathcal C_k^{(n)} = \max \left \{ s \mathcal C_ {k,\mu} ^{(n)} , s \mathcal C_ {k,\nu} ^{(n)} \right \}$.

To simplify the search for the boundaries, note that for any 4 or 5 qubit CSX-state, the subconcurrence terms (26-33) are strictly increased by setting the coefficient phases to 0.  This implies that the boundaries can be searched for among 4 and 5 qubit CSX states with purely real coefficients.

\subsection{4 Qubits}

Consider an arbitrary 4 qubit CSX state, (21), with real coefficients.  The corresponding normalized state
\begin{equation}
\ket{\bar{\psi}}=\frac{1}{\sqrt{a^2+c^2+f^2}} \left( a \ket{0000} + c \overbrace{\ket{1100}} + f \ket{1111} \right),
\end{equation}
has both larger or equal $s\mathcal C^{(4)}_{1,\nu}$ and larger or equal $s\mathcal C^{(4)}_{2,\mu}$.  To show this, consider either subconcurrence, $s\mathcal{C}$, for which it is then true that
\begin{equation}
s\mathcal{C}\left(\ket{\bar{\psi}} \right) = s\mathcal{C}\left(\ket{\psi}  \right) \biggr{|}_{d=0} \geq s\mathcal{C}\left(\ket{\psi}   \right).
\end{equation}
All of which implies that the boundary of the $\left( s\mathcal C^{(4)}_{1,\nu}, s\mathcal C^{(4)}_{2,\mu} \right)$ pairs can be looked for among states with $d=0$.  Likewise the state
\begin{equation}
\ket{\bar {\psi}} = \sqrt {\frac {a^2+f^2}{2}} \left( \ket{0000} + \ket{1111} \right) + c \overbrace{\ket{1100}} + d \overbrace{\ket{1010}}
\end{equation}
has larger or equal $s\mathcal{C}_{1,\mu}^{(4)}$, $s\mathcal{C}_{1,\nu}^{(4)}$, and $s\mathcal{C}_{2,\nu}^{(4)}$, meaning the boundaries of the  $\left (s\mathcal{C}_{1,\mu}^{(4)}, s\mathcal{C}_{2,\nu}^{(4)} \right)$ and $\left (s\mathcal{C}_{1,\nu}^{(4)}, s\mathcal{C}_{2,\nu}^{(4)} \right)$ pairs can be found among states where $a=f$.  And lastly the state
\begin{equation}
\ket{\bar {\psi}} = \frac{1}{\sqrt{c^2-d^2}} \left( c \overbrace{\ket{1100}} + d \overbrace{\ket{1010}} \right)
\end{equation}
has larger or equal $s\mathcal{C}_{1,\mu}^{(4)}$ and $s\mathcal{C}_{2,\mu}^{(4)}$, so the boundary of the $\left( s\mathcal{C}_{1,\mu}^{(4)}, s\mathcal{C}_{2,\mu}^{(4)}\right)$ pairs can be found among states where $a=f=0$.

Using these simplified states, the remaining coefficients can be expressed using the following spherical parametrizations,
\begin{eqnarray}
\left \{ a,c,f \right \} &\rightarrow& \left \{ \sin \theta \cos \phi , \cos \theta, \sin \theta \sin \phi \right \} \\ 
\left \{ a,c,d \right \} &\rightarrow& \left \{ \cos \alpha, \sin \alpha \cos \beta, \sin \alpha \sin \beta \right \} \\
\left \{ c,d \right \} &\rightarrow& \left \{ \cos \zeta , \sin \zeta \right \},
\end{eqnarray}
associated with the $\left( s\mathcal C^{(4)}_{1,\nu}, s\mathcal C^{(4)}_{2,\mu} \right)$, $\left( s\mathcal C^{(4)}_{1,\mu (\nu) }, s\mathcal C^{(4)}_{2,\nu} \right)$, and $\left( s\mathcal C^{(4)}_{1,\mu}, s\mathcal C^{(4)}_{2,\mu} \right)$ pairs respectively, where $\left \{ \theta, \phi, \alpha, \beta, \zeta \right \} \in [0,\pi/2]$.  In these parametrizations, we can define the maps,
\begin{eqnarray}
\mathcal{C}_{\nu,\mu} &:& \left \{ \theta, \phi \right \} \rightarrow \left \{ s\mathcal C_{1,\nu}^{(4)} , s\mathcal C_{2,\mu}^{(4)} \right \} \\
\mathcal{C}_{\mu (\nu),\nu} &:& \left \{ \alpha, \beta \right \} \rightarrow \left \{ s\mathcal C_{1,\mu(\nu)}^{(4)} , s\mathcal C_{2,\nu}^{(4)} \right \} \\
\mathcal{C}_{\mu,\mu} &:& \zeta \rightarrow \left \{ s\mathcal C_{1,\mu}^{(4)} ,s \mathcal C_{2,\mu}^{(4)} \right \}
\end{eqnarray}
according to the expressions (26-29).  The boundaries of the images of these maps correspond to the boundaries of the domains, as well as the zeroes of the determinant of the Jacobians for each map.  The result of these boundary determinations leaves the following two outermost boundaries,
\begin{widetext}
\begin{equation}
s\mathcal{C}_{2,X}^{(4)} \leq \begin{cases}
    \frac 25 \left( 8 \sqrt{1-2s\mathcal{C}_{1,X}^{(4)}-4\left(s\mathcal{C}_{1,X}^{(4)}\right)^2} - s\mathcal{C}_{1,X}^{(4)} +1 \right) , & -\frac 12 \leq s\mathcal{C}_{1,X}^{(4)} \leq \frac{63}{226}\\
    \frac 19 \left( 8 \sqrt{1-s\mathcal{C}_{1,X}^{(4)}-2\left(s\mathcal{C}_{1,X}^{(4)}\right)^2} - 4s\mathcal{C}_{1,X}^{(4)} -1 \right), & \frac{63}{226} \leq s\mathcal{C}_{1,X}^{(4)} \leq \frac 12,
  \end{cases}
\end{equation}
\end{widetext}
which came from the $\left( s\mathcal C^{(4)}_{1,\nu}, s\mathcal C^{(4)}_{2,\mu} \right)$ pairs.  These boundaries are displayed in Figure 4.
\begin{center}
\begin{figure}[h]
\centering
\includegraphics[width = 85mm]{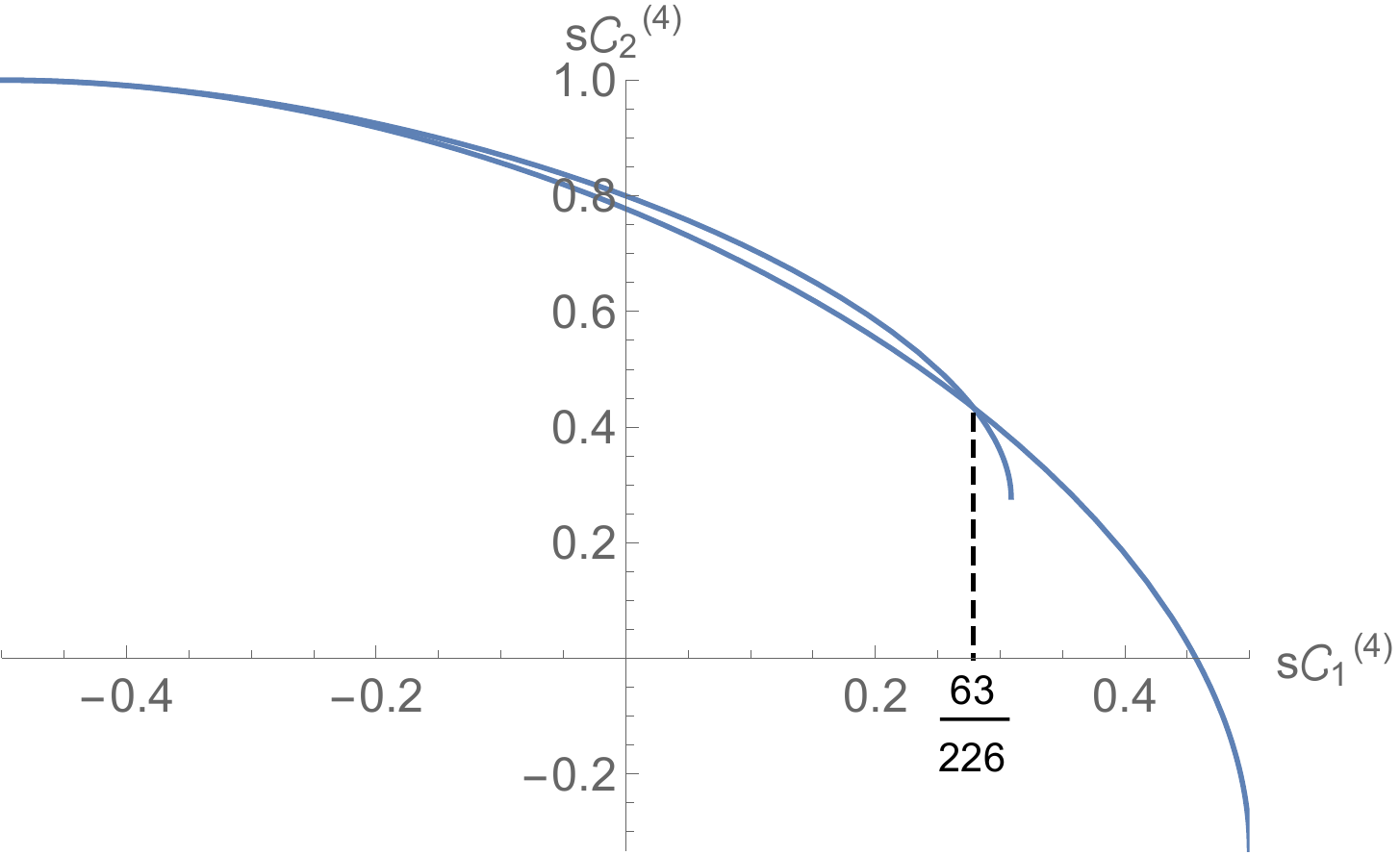}
\caption{The 4 qubit CSX state subconcurrence boundaries.}
\end{figure}
\end{center}

\subsection{5 Qubits}

Following the methods from the previous section, start by considering an arbitrary 5 qubit CSX state, (22), with real coefficients.  The corresponding normalized state,
\begin{equation}
\ket{\bar{\psi}} = \frac{1}{\sqrt{c^2+d^2+g^2}} \left( c \overbrace{\ket{00011}} + d \overbrace{\ket{00101}} + g \overbrace{\ket{01111}} \right),
\end{equation}
has larger or equal $s\mathcal{C}_{1,\mu}^{(5)}$, and $s\mathcal{C}_{2,\mu}^{(5)}$, so therefore the boundary of the $\left( s\mathcal{C}_{1,\mu}^{(5)}, s\mathcal{C}_{2,\mu}^{(5)}\right) $ pairs can be searched for among states with $a=0$.  For the other pairs, we will bound their subconcurrences by a sequence of lines which lie within the $\left( s\mathcal{C}_{1,\mu}^{(5)}, s\mathcal{C}_{2,\mu}^{(5)}\right) $ boundary.

We can now parametrize the remaining coefficients of (50) as
\begin{equation}
\left \{c,d,g \right \} \rightarrow \left \{ \sin \theta \cos \phi , \sin \theta \sin \phi, \cos \theta \right \},
\end{equation}
and define the map
\begin{equation}
\mathcal C _{\mu, \mu}: \left \{ \theta , \phi \right \} \rightarrow \left \{ s\mathcal C_{1,\mu}^{(5)},s\mathcal C_{2,\mu}^{(5)} \right \},
\end{equation}
according to (30) and (32).  By analyzing the boundaries of the domain and the zeroes of the determinant of the Jacobian of this map, three boundaries make up a maximal set, as plotted in Figure 5.  These three boundaries are parametrized by $\theta=\frac{\pi}2$, $\phi=0$, and $\phi=\frac{\pi}2$.  The exact polynomials in $s\mathcal C_{1,X}^{(5)}$ and $s\mathcal C_{2,X}^{(5)}$ which describe these boundaries are easily determined by a Gr{\"o}bner basis calculation performed on (52), but the results are quite lengthy.
\begin{center}
\begin{figure}[h]
\centering
\includegraphics[width = 85mm]{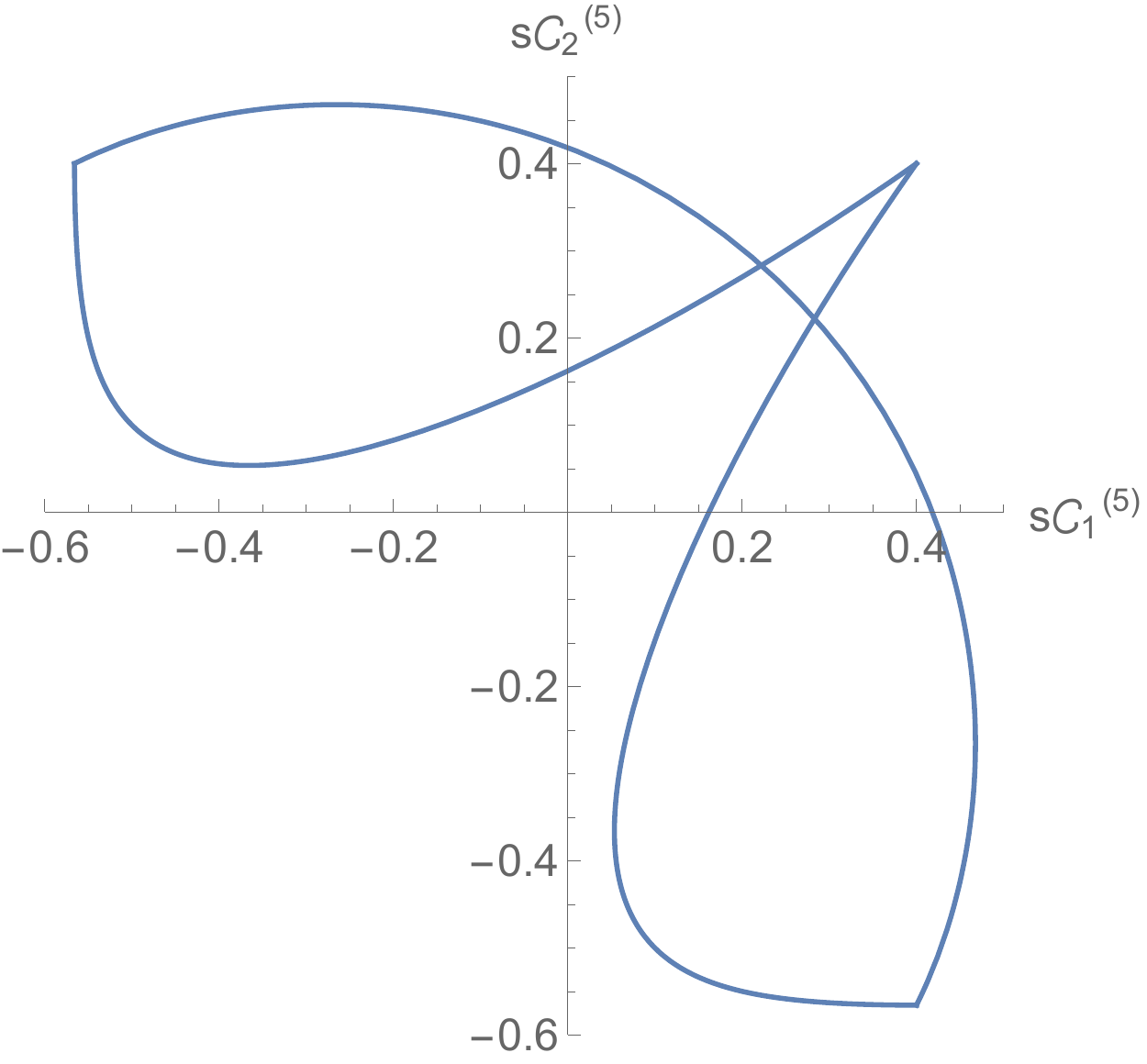}
\caption{The 5 qubit CSX state subconcurrence boundaries.}
\end{figure}
\end{center}

Turning now to the remaining subconcurrence pairings.  It was shown in Table 1 that $\mathcal{C}_{1(2),\nu}^{(5)} \leq 0.366$.  Another simple maximization shows that $s\mathcal{C}_{1,\nu}^{(5)} + s\mathcal{C}_{2,\nu}^{(5)} \leq \frac 25$.  These three conditions bound the $\left( s\mathcal{C}_{1,\nu}^{(5)}, s\mathcal{C}_{2,\nu}^{(5)}\right) $ pairs to a region well within the previous boundary, as shown in Figure 6.
\begin{center}
\begin{figure}[h]
\centering
\includegraphics[width = 85mm]{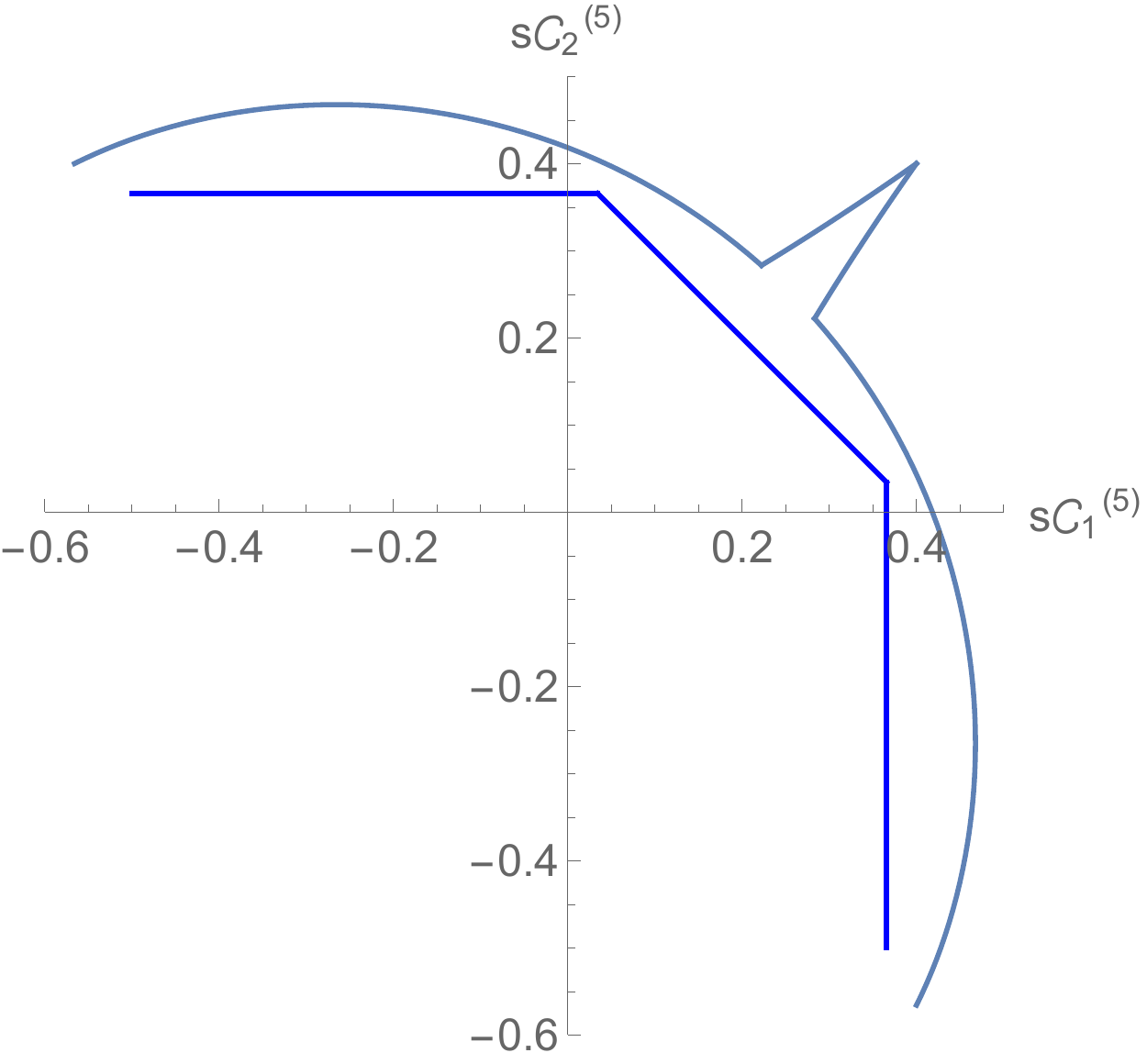}
\caption{The bounding conditions on the $\left( s\mathcal{C}_{1,\nu}^{(5)}, s\mathcal{C}_{2,\nu}^{(5)}\right) $ pairs (darker blue) with the overall boundary .}
\end{figure}
\end{center}

Lastly, the remaining two pairs, $\left( s\mathcal{C}_{1,\mu (\nu)}^{(5)}, s\mathcal{C}_{2,\nu (\mu)}^{(5)}\right) $ can be handled together due to the symmetry in 5 qubits.  Similarly to the previous pair boundary, we will find a set of lines which bound the $\left( s\mathcal{C}_{1,\mu}^{(5)}, s\mathcal{C}_{2,\nu}^{(5)}\right) $ pairs.  We can again take advantage of $s\mathcal{C}_{2,\nu}^{(5)} \leq 0.366$, as well as the following two new maximizations,
\begin{eqnarray}
s\mathcal{C}_{2,\nu}^{(5)}+s\mathcal{C}_{1,\mu}^{(5)} &\leq& \frac {47}{100} \\
s\mathcal{C}_{2,\nu}^{(5)}+2 s\mathcal{C}_{1,\mu}^{(5)} &\leq& \frac 45 .
\end{eqnarray}
These three conditions bound the $\left( s\mathcal{C}_{1,\mu}^{(5)}, s\mathcal{C}_{2,\nu}^{(5)}\right) $ pairs within the original boundary for $0 \leq s\mathcal{C}_{1,\mu}^{(5)}$ and $0 \leq  s\mathcal{C}_{2,\nu}^{(5)} $, as shown in Figure 7.
\begin{center}
\begin{figure}[h]
\centering
\includegraphics[width = 85mm]{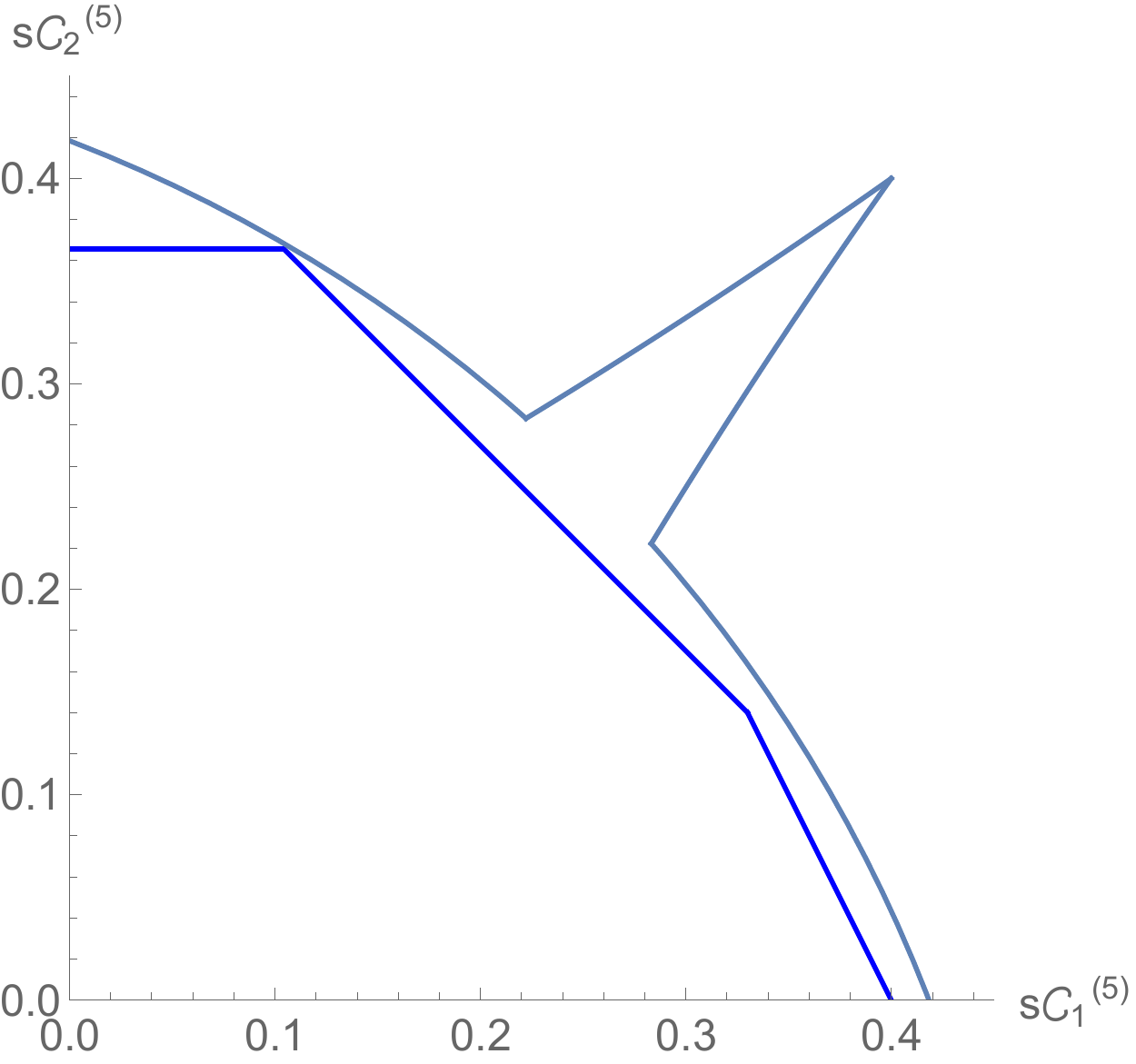}
\caption{The bounding conditions on the $\left( s\mathcal{C}_{1,\mu}^{(5)}, s\mathcal{C}_{2,\nu}^{(5)}\right) $ pairs (darker blue) with the overall boundary .}
\end{figure}
\end{center}
Note that these conditions on the $\left( s\mathcal{C}_{1,\mu}^{(5)}, s\mathcal{C}_{2,\nu}^{(5)}\right) $ do not actually fall within the original boundary for regions in $0.4 \leq s\mathcal{C}_{1,\mu}^{(5)}$ and $  s\mathcal{C}_{2,\nu}^{(5)} \leq 0$.  But given that $  s\mathcal{C}_{2,\nu}^{(5)} \leq 0$ for that region, the actual concurrences would be mapped to $\mathcal C_2^{(5)} = 0$, where the boundaries would then agree.

\end{document}